\documentclass[12pt]{amsart}
\usepackage{amssymb}
\usepackage{color}
\usepackage{amsmath,epic,curves,amscd}
\usepackage[english]{babel}
\usepackage{graphicx}
\usepackage{comment}
\usepackage{appendix}
\usepackage{mathdots}
\usepackage[all]{xy}
\newtheorem{claim}{}[section]
\newtheorem{theorem}[claim]{Theorem}

\newtheorem{proposition}[claim]{Proposition}

\theoremstyle{remark}

\renewenvironment{proof}{\noindent{\it Proof. \hskip0pt}}
                      {$\square$\par\medskip}

\begin{document}
\baselineskip 6.0 truemm
\parindent 1.5 true pc

\newcommand\lan{\langle}
\newcommand\ran{\rangle}
\newcommand\tr{{\text{\rm Tr}}\,}
\newcommand\ot{\otimes}
\newcommand\ol{\overline}
\newcommand\join{\vee}
\newcommand\meet{\wedge}
\renewcommand\ker{{\text{\rm Ker}}\,}
\newcommand\image{{\text{\rm Im}}\,}
\newcommand\id{{\text{\rm id}}}
\newcommand\tp{{\text{\rm tp}}}
\newcommand\pr{\prime}
\newcommand\e{\epsilon}
\newcommand\la{\lambda}
\newcommand\inte{{\text{\rm int}}\,}
\newcommand\ttt{{\text{\rm t}}}
\newcommand\spa{{\text{\rm span}}\,}
\newcommand\conv{{\text{\rm conv}}\,}
\newcommand\rank{\ {\text{\rm rank of}}\ }
\newcommand\re{{\text{\rm Re}}\,}
\newcommand\ppt{\mathbb T}
\newcommand\rk{{\text{\rm rank}}\,}
\newcommand\SN{{\text{\rm SN}}\,}
\newcommand\SR{{\text{\rm SR}}\,}
\newcommand\HA{{\mathcal H}_A}
\newcommand\HB{{\mathcal H}_B}
\newcommand\HC{{\mathcal H}_C}
\newcommand\CI{{\mathcal I}}
\newcommand{\bra}[1]{\langle{#1}|}
\newcommand{\ket}[1]{|{#1}\rangle}
\newcommand\cl{\mathcal}
\newcommand\idd{{\text{\rm id}}}
\newcommand\OMAX{{\text{\rm OMAX}}}
\newcommand\OMIN{{\text{\rm OMIN}}}
\newcommand\diag{{\text{\rm Diag}}\,}
\newcommand\calI{{\mathcal I}}
\newcommand\bfi{{\bf i}}
\newcommand\bfj{{\bf j}}
\newcommand\bfk{{\bf k}}
\newcommand\bfl{{\bf l}}
\newcommand\bfp{{\bf p}}
\newcommand\bfq{{\bf q}}
\newcommand\bfzero{{\bf 0}}
\newcommand\bfone{{\bf 1}}

\title{Separability of three qubit Greenberger-Horne-Zeilinger diagonal states}

\author{Kyung Hoon Han and Seung-Hyeok Kye}
\address{Department of Mathematics, The University of Suwon, Gyeonggi-do 445-743, Korea}
\email{kyunghoon.han at gmail.com}
\address{Department of Mathematics and Institute of Mathematics, Seoul National University, Seoul 151-742, Korea}
\email{kye at snu.ac.kr}
\thanks{KHH and SHK were partially supported by NRF-2012R1A1A1012190 and 2015R1D1A1A02061612, respectively.}



\begin{abstract}
We characterize the separability of three qubit GHZ diagonal states
in terms of entries. This enables us to check separability of GHZ diagonal states
without decomposition into the sum of pure product states. In the course of discussion,
we show that the necessary criterion of G\"uhne
\cite{guhne_pla_2011} for (full) separability of three qubit GHZ
diagonal states is sufficient with a simpler formula.
The main tool is to use entanglement witnesses
which are tri-partite Choi matrices of positive bi-linear maps.
\end{abstract}

\maketitle

\section{Introduction}

Entanglement is now considered as one of the most important
resources in quantum information theory, and it is crucial to detect
entanglement from separability. Positivity of partial transposes is a simple but powerful
criterion for separability \cite{{choi-ppt},{peres}}.
In fact, it is known \cite{{woronowicz},{horo-1}} that PPT property
is equivalent to separability for $2\ot 2$ or $2\ot 3$ systems.
But, it is very difficult in general to distinguish separability from entanglement,
as it is known to be an $NP$-hard problem \cite{gurvits}.

The purpose of this note is to give a complete characterization of separability
for three qubit Greenberger-Horne-Zeilinger diagonal states, which are diagonal in the GHZ basis.
Those include mixtures of GHZ states and identity, as they were considered in \cite{{dur},{murao}} for examples.
In multi-qubit systems,
GHZ states \cite{GHSZ, GHZ} are key examples of maximally entangled states,
and they are known to have many applications in quantum information theory.
They also play central roles in the classification of multi-qubit entanglement \cite{{dur},{abls},{dur_multi}}.
See survey articles \cite{{guhne_survey},{horo-survey}}
for general theory of entanglement as well as various aspects of GHZ states.

Our main tool is to use the notion of entanglement witnesses.
In the bi-partite cases, positive linear maps are very
useful to detect entanglement through the duality between tensor products
and linear maps \cite{horo-1,eom-kye}. This duality has been
formulated as the notion of entanglement witnesses  \cite{terhal}, which is still
valid in multi-partite cases. In the tri-partite cases, the second
author \cite{kye_multi_dual} has interpreted entanglement witnesses
as positive bi-linear maps. With this interpretation, we carefully choose
useful entanglement witnesses for our purpose.

There were two major steps to characterize separability of three
qubit GHZ diagonal states in the literature. Kay \cite{kay_2011}
gave a condition under which a GHZ diagonal state is separable if
and only if it is of positive partial transpose. On the other hand,
G\"uhne \cite{guhne_pla_2011} gave a necessary criterion for
separability of three qubit GHZ diagonal states, and provided a
sufficient condition for separability.
We will describe their results after we give definitions
and fix notations in the next section, where we will also
explain what is our work precisely.

\section{GHZ diagonal states}

We recall that a state is (fully) separable if it is a
convex combination of pure product states, and entangled if it is not separable.
So, a three qubit state $\varrho$ in the tensor product $M_2\ot M_2\ot M_2$ is separable if and only if
it can be written by the finite sum of the form
$$
\varrho=\sum_i p_i |\xi_i\ran \ot |\eta_i\ran \ot |\zeta_i\ran\lan \xi_i|\ot \lan\eta_i|\ot \lan \zeta_i|
$$
with $p_i>0$, $\sum_i p_i=1$ and two dimensional vectors $|\xi_i\ran, |\eta_i\ran$ and  $|\zeta_i\ran$.
Throughout this paper, three qubit states will be considered as $8\times 8$ matrices with the identification
$M_2\otimes M_2\otimes M_2 = M_8$ through the lexicographic order
of indices in the tensor product. We may take some subsystems and take the transposes for them to get partial transposes.
For example, we take the first system, then the corresponding partial transpose of a three qubit product state
$x\ot y\ot z$ is given by $x^\ttt\ot y\ot z$. 
The PPT criteria \cite{{choi-ppt},{peres}} tell us that if $\varrho$ is separable
then all the partial transposes of $\varrho$ are positive.

The three qubit GHZ state basis consists of eight vectors in $\mathbb C^2\otimes\mathbb C^2\otimes\mathbb C^2$ given by
$$
|\xi_{ijk}\ran=\frac 1{\sqrt 2}\left( |i\ran \ot|j\ran \ot|k\ran
+(-1)^i|\bar i\ran  \ot|\bar j \ran \ot|\bar k\ran \right),\qquad i+\bar i=1\mod 2,
$$
where the index $ijk$ runs through $i,j,k\in\{0,1\}$. We endow the
indices with the lexicographic order to get eight vectors
$\xi_1,\xi_2,\dots,\xi_8$. A GHZ diagonal state is of the form
$$
\varrho=\sum_{i=1}^8 p_i|\xi_i\ran\lan \xi_i|
$$
for nonnegative $p_i$'s with $\sum_{i=1}^8 p_i=1$. Then we see that
$2\varrho$ is written as the following $8\times 8$ matrix:
$$
X(a,b,c):=
\left(
\begin{matrix}
a_1 &&&&&&& c_1\\
& a_2 &&&&& c_2 & \\
&& a_3 &&& c_3 &&\\
&&& a_4&c_4 &&&\\
&&& \bar c_4& b_4&&&\\
&& \bar c_3 &&& b_3 &&\\
& \bar c_2 &&&&& b_2 &\\
\bar c_1 &&&&&&& b_1
\end{matrix}
\right),
$$
whose entries are all zero except for diagonal and anti-diagonal entries with
$$
\begin{aligned}
a=b&=(p_1+p_8,\ p_2+p_7,\ p_3+p_6,\ p_4+p_5)\in\mathbb R^4,\\
c&=(p_1-p_8,\ p_2-p_7,\ p_3-p_6,\ p_4-p_5)\in\mathbb R^4.
\end{aligned}
$$

The self-adjoint matrix of the form $X(a,b,c)$ with $a,b\in\mathbb R^4$ and $c\in\mathbb C^4$
is called {\sf X}-shaped. Therefore, we have seen that a three qubit GHZ diagonal state
is an {\sf X}-shaped state, or {\sf X}-state in short, $X(a,b,c)$ with $a,b\in\mathbb R^4$ and $c\in\mathbb R^4$.
Conversely, every {\sf X}-state $X(a,b,c)$ can be realized as a GHZ diagonal state
whenever $a=b$ and $c\in\mathbb R^4$.
A mixture of $|\xi_1\ran\lan \xi_1|$ and the identity is the simplest example
of a nontrivial GHZ diagonal state, as it was considered in \cite{murao}.
D\" ur, Cirac and Tarrach \cite{dur} considered GHZ diagonal states
with $c=(c_1,0,0,0)$, and showed that these states are separable if and only if they are of PPT.
See also \cite{dur_multi} for multi-qubit analogues.
On the other hand, {\sf X}-states with $a_ib_i=1$ and $c=(1,0,0,0)$
were considered \cite {abls} in the contexts of classification of three qubit PPT entangled edge states.
See also \cite{guhne10} for criteria for separability of three qubit states by their {\sf X}-parts.

There are other expression \cite{schack} of GHZ diagonal states, using Pauli matrices
$$
X=\begin{pmatrix}0&1\\1&0\end{pmatrix},\quad
Y=\begin{pmatrix}0&-{\rm i}\\{\rm i}&0\end{pmatrix},\quad
Z=\begin{pmatrix}1&0\\0&-1\end{pmatrix}.
$$
If $a=b$ and $c\in\mathbb R^4$, then a GHZ diagonal state
$\varrho=X(a,a,c)$ can be written as
$$
\begin{aligned}
X(a,a,c) = &{1 \over 8} (I \otimes I \otimes I
+ \lambda_2 Z \otimes Z \otimes I
+ \lambda_3 Z \otimes I \otimes Z
+ \lambda_4 I \otimes Z \otimes Z\\
&+ \lambda_5 X \otimes X \otimes X
+ \lambda_6 Y \otimes Y \otimes X
+ \lambda_7 Y \otimes X \otimes Y
+ \lambda_8 X \otimes Y \otimes Y),
\end{aligned}
$$
with the coefficients
\begin{equation}\label{lambdas}
\begin{aligned}
\lambda_2 &= 2(+a_1+a_2-a_3-a_4),\\
\lambda_3 &=2(+a_1-a_2+a_3-a_4),\qquad \lambda_4=2(+a_1-a_2-a_3+a_4),\\
\lambda_5 &= 2(+c_1+c_2+c_3+c_4),\qquad \lambda_6 = 2(-c_1-c_2+c_3+c_4),\\
\lambda_7 &= 2(-c_1+c_2-c_3+c_4),\qquad \lambda_8 = 2(-c_1+c_2+c_3-c_4).
\end{aligned}
\end{equation}
Now, we are ready to describe the results on the separability of
a three qubit state $\varrho=X(a,a,c)$ by Kay \cite{kay_2011} and G\"uhne \cite{guhne_pla_2011}:
\begin{itemize}
\item
If $\Pi_{i=5}^8 \lambda_i \le 0$, then $\varrho$ is separable if and only if $\varrho$ is of PPT \cite{kay_2011}.
\item
If $\Pi_{i=5}^8 \lambda_i >0$, then
the inequality
\begin{equation}\label{sufficient}
\min \{a_1,a_2,a_3,a_4\} \ge {\sqrt{(\lambda_5\lambda_6+\lambda_7\lambda_8)
(\lambda_5\lambda_7+\lambda_6\lambda_8)
(\lambda_5\lambda_8+\lambda_6\lambda_7)}
\over
{8\sqrt{\lambda_5\lambda_6\lambda_7\lambda_8}}}
\end{equation}
implies the separability of $\varrho$ \cite{guhne_pla_2011}.
\end{itemize}
Kay \cite{kay_2011} also considered the state
$\varrho=X(a,a,c)$ with $a=(4+\alpha,\alpha,\alpha,\alpha)$ and $c=(2,2,-2,2)$
to give examples of PPT entangled states among GHZ diagonal states. Our main contribution is
to give a condition in terms of anti-diagonal entries of $\varrho$ with the following properties:
\begin{itemize}
\item
If $\varrho$ satisfies this condition, then $\varrho$ is separable if and only if $\varrho$ satisfies (\ref{sufficient}).
\item
If $\varrho$ does not meet this condition, then $\varrho$ is separable if and only if $\varrho$ is of PPT.
\end{itemize}
See Theorem \ref {main_condition} for the precise description of this condition.
This completes the characterization of separability of three qubit GHZ diagonal states
in terms of entries. This is one of few cases in the literature where the separability problem
is solved in terms of entries. With this characterization, we may confirm the separability of a GHZ diagonal state
without a concrete decomposition into the sum of pure product states, which is usually a quite nontrivial job.

Suppose that the {\sf X}-part of a three qubit state $\varrho$ is given by $X(a,b,c)$
with $a,b\in\mathbb R^4$ and $c\in\mathbb C^4$.
In order to get
a necessary condition for separability
of GHZ diagonal states, G\"uhne \cite{guhne_pla_2011} also introduced
the following:
$$
\begin{aligned}
\mathcal L (\varrho, z) := & {\rm Re} \left(z_1 c_1 + z_2 c_2 + z_3 c_3 + z_4 \bar c_4 \right), \qquad z\in\mathbb C^4\\
\mathcal F(z) := & {\rm Re} (z_1) \cos (\alpha+\beta+\gamma) - {\rm Im} (z_1) \sin (\alpha+\beta+\gamma)
+ {\rm Re} (z_2) \cos (\alpha) - {\rm Im} (z_2) \sin (\alpha) \\
&  + {\rm Re} (z_3) \cos (\beta) - {\rm Im} (z_3) \sin (\beta)
+ {\rm Re} (z_4) \cos (\gamma) - {\rm Im} (z_4) \sin (\gamma), \\
C(z):= & \sup_{\alpha,\beta,\gamma} |\mathcal F (z)|
\end{aligned}
$$
and showed that if $\varrho=X(a,b,c)$ is separable then the inequality
\begin{equation}\label{Guhne}
\mathcal L (\varrho, z) \le C(z)~ \Delta_\varrho
\end{equation}
holds for every $z \in \mathbb C^4$, where
the number $\Delta_\varrho$ is given by
$$
\Delta_\varrho
=\min \left\{\sqrt{a_ib_i}\ (i=1,2,3,4),\ \sqrt[4]{a_1b_2b_3a_4},\ \sqrt[4]{b_1a_2a_3b_4} \right\}
$$
which is determined by the diagonal entries of $\varrho=X(a,b,c)$. If $\varrho$ is GHZ diagonal,
then we have $\Delta_\varrho =\min\{a_1, a_2, a_3, a_4\}$.
The number $C(z)$ above turns out to coincide essentially with the number in the characterization
of three qubit {\sf X}-shaped entanglement witnesses \cite{han_kye_tri}.
In this way, we got a simpler expression for $C(z)$. See Proposition \ref{3to1}.

In the next section, we will also consider {\sf X}-shaped entanglement witnesses, and show that
special kinds of them are enough to confirm the separability of GHZ diagonal states.
These will be used in Section 4 to prove our result
that G\"uhne's necessary condition is also sufficient for separability of three qubit GHZ diagonal states. In Section 5, we apply this condition
to give a concrete entry-wise characterization of separability for three qubit GHZ diagonal states.

\section{Entanglement witnesses}

We recall that a non-positive self-adjoint matrix $W$ in $M_A\otimes M_B\otimes M_C$ is an entanglement witness if
$$
\lan \varrho, W\ran:=\tr (\varrho W^\ttt)\ge 0
$$
for every separable state $\varrho$. A matrix $W$ in $M_A\otimes M_B\otimes M_C$ is written by
$$
\begin{aligned}
W
&=\sum_{i_1,j_1}|i_1\ran\lan j_1|\ot W_{i_1,j_1}\in M_{A}\ot (M_B\otimes M_C)\\
&=\sum_{i_1,j_1}\sum_{i_2,j_2}|i_1\ran\lan j_1|\otimes |i_2\ran\lan j_2|\ot W_{i_1i_2,j_1j_2}\in M_A\otimes M_B\otimes M_C
\end{aligned}
$$
in a unique way, and so we may associate a bi-linear map $\phi_W:M_A\times M_B\to M_C$ by
$$
\phi_W( |i_1\ran\lan j_1|, |i_2\ran\lan j_2|)= W_{i_1i_2,j_1j_2}
$$
for matrix units $\{|i_1\ran\lan j_1|\}$ and $\{|i_2\ran\lan j_2|\}$ of $M_A$ and $M_B$, respectively.
It is known \cite{kye_multi_dual} that $\lan \varrho, W\ran\ge 0$ for every separable state $\rho$ if and
only if $\phi_W$ is a positive bi-linear map, that is, $\phi(x,y)\in M_C$ is positive whenever $x\in M_A$
and $y\in M_B$ are positive.

For a given $|x\ran=(x_0,x_1)^\ttt\in\mathbb C^2$, we write $|x_\pm\ran =(x_0,\pm x_1)^\ttt$.
Then the {\sf X}-part of the pure product state
$\varrho=|\xi\ran\lan\xi|$ with $|\xi\ran=|x\ran\otimes |y\ran\otimes |z\ran$
is given by the average of the four pure product states:
$\varrho_X=\frac 14 \sum_{k=0}^3 |\xi_k\ran\lan\xi_k|$,
where $|\xi_k\ran$ is given by
$$
\begin{aligned}
|\xi_0\ran=&|x_+ \ran \ot |y_+ \ran \ot |z_+ \ran,\\
|\xi_1\ran=&|x_+ \ran \ot |y_- \ran \ot |z_- \ran,\\
|\xi_2\ran=&|x_- \ran \ot |y_+ \ran \ot |z_- \ran,\\
|\xi_3\ran=&|x_- \ran \ot |y_- \ran \ot |z_+ \ran.
\end{aligned}
$$
Therefore, we see that the {\sf X}-part of a separable state is again separable. This simple observation has
an important implication.

\begin{proposition}\label{X-witness}
The {\sf X}-part of a three qubit entanglement witness is still an entanglement witness unless it is positive.
\end{proposition}

\begin{proof}
We denote by $W_X$ and $\varrho_X$ the {\sf X}-parts of an entanglement witness $W$ and a state $\varrho$, respectively.
Suppose that $\varrho$ is separable. Then we have
$\lan \varrho, W_X\ran = \lan\varrho_X, W\ran\ge 0$,
because $\varrho_X$ is separable. This shows that $W_X$ is an entanglement witness.
\end{proof}

We say that an {\sf X}-shaped self-adjoint matrix $W=X(s,t,u)$ is {\sl GHZ diagonal} if $s=t$ and $u \in \mathbb R$.

\begin{theorem}\label{GHZ-witness}
Let $\varrho$ be a GHZ diagonal state. Then, $\varrho$ is separable if and only if
$\lan \varrho, W \ran \ge 0$ for every GHZ diagonal witness $W$.
\end{theorem}

\begin{proof}
The `only if' part is clear. For the `if' part, it suffices to show that $\lan \varrho, W \ran \ge 0$
for every {\sf X}-shaped witness $W$ by Proposition \ref{X-witness}.
We consider the operation on $M_2 \otimes M_2 \otimes M_2$ which interchanges $|0\ran$ and $|1\ran$ in each subsystem.
This operation sends an entanglement witness $W=X(s,t,u)$ to the entanglement witness
$\widetilde W=X(t,s,\bar u)$. Explicitly, we can write $\widetilde W = U W U^*$ for the symmetric unitary
$$
U=\begin{pmatrix}0&1\\1&0\end{pmatrix} \otimes \begin{pmatrix}0&1\\1&0\end{pmatrix} \otimes \begin{pmatrix}0&1\\1&0\end{pmatrix}.
$$
The average
$$
{W + \widetilde W \over 2} = X\left({s+t \over 2}, {s+t \over 2}, {\rm Re} ~u\right)
$$
is GHZ diagonal, and $W$ is GHZ diagonal if and only if $(W + \widetilde W) \slash 2=W$. In other words,
the mapping $W \mapsto (W + \widetilde W) \slash 2$ is an idempotent from the real space of {\sf X}-shaped self-adjoint matrices
onto its subspace consisting of GHZ diagonal matrices.
We have
$$
\begin{aligned}
\lan \varrho, W \ran & = \left\lan {\varrho + \widetilde \varrho \over 2}, W \right\ran \\
& = {1 \over 2} \left(\lan \varrho, W\ran + \tr (U\varrho U^* W^\ttt)\right) \\
& = {1 \over 2} \left(\lan \varrho, W\ran + \tr (\varrho (U W U^*)^\ttt)\right)
 = \left\lan \varrho, {W + \widetilde W \over 2}\right\ran
 \ge 0,
\end{aligned}
$$
which completes the proof.
\end{proof}

For a three qubit {\sf X}-shaped self-adjoint matrix $W=X(s,t,u)$, the authors \cite{han_kye_tri} have shown
that $W=X(s,t,u)$ is an entanglement witness if and only if it is non-positive and satisfies
the inequality
$$
\begin{aligned}
\sqrt{(s_1+t_4|\alpha|^2)(s_4+t_1|\alpha|^2)}
+&\sqrt{(s_2+t_3|\alpha|^2)(s_3+t_2|\alpha|^2)}\\
&\phantom{XXXX}\ge |u_1\bar\alpha +\bar u_4 \alpha|+|u_2 \bar \alpha + \bar u_3 \alpha|
\end{aligned}
$$
for each complex number $\alpha\in\mathbb C$.
If we write $\alpha=r e^{{\rm i}\theta}$, then the above inequality
is equivalent to the following
$$
\begin{aligned}
\sqrt{(s_1 r^{-1}+t_4r)(s_4r^{-1}+t_1r)}
+&\sqrt{(s_2r^{-1}+t_3r)(s_3r^{-1}+t_2r)}\\
&\phantom{XXXX}\ge |u_1e^{-{\rm i}\theta} +\bar u_4 e^{{\rm i}\theta}|+|u_2 e^{-{\rm i}\theta} + \bar u_3 e^{{\rm i}\theta}|,
\end{aligned}
$$
which holds for every $r>0$ and $\theta$. For a given $(s,t)\in\mathbb R^4_+\times\mathbb R^4_+$ and $u\in\mathbb C^4$, we introduce
two numbers:
$$
\begin{aligned}
A(s,t)
&=\inf_{r>0} \left[\sqrt{(s_1 r^{-1}+t_4r)(s_4r^{-1}+t_1r)} +\sqrt{(s_2r^{-1}+t_3r)(s_3r^{-1}+t_2r)}\right],\\
B(u)
&=\max_\theta\left(|u_1e^{{\rm i}\theta} +\bar u_4|+|u_2 e^{{\rm i}\theta} + \bar u_3|\right).
\end{aligned}
$$
Here, $\mathbb R_+$ denotes the interval $[0,\infty)$. Then we have the following:

\begin{proposition}\label{XEW}
A three qubit non-positive self-adjoint matrix $W=X(s,t,u)$ is an entanglement witness if and only if the inequality
$A(s,t)\ge B(u)$ holds.
\end{proposition}

It is surprising to note that the number $B(u)$ is essentially identical with the number $C(z)$ in the G\"uhne's criterion.
This enables us to calculate the number $C(z)$ in terms of entries of $z$.

\begin{proposition}\label{3to1}
For $z \in \mathbb C^4$, we have
$C(z) = B(z_1,z_2,z_3,\bar z_4)$.
\end{proposition}

\begin{proof}
We have
$$
\begin{aligned}
\mathcal F(z) & = {\rm Re} \left(z_1 e^{{\rm i}(\alpha+\beta+\gamma)} + z_2 e^{{\rm i}\alpha} + z_3 e^{{\rm i}\beta} + z_4 e^{{\rm i}\gamma} \right)\\
& = {\rm Re} \left(z_1 e^{{\rm i}(\alpha+\beta+\gamma)} + z_2 e^{{\rm i}\alpha} + \bar z_3 e^{-{\rm i}\beta} + z_4 e^{{\rm i}\gamma} \right)\\
& \le |z_1 e^{{\rm i}(\alpha+\beta+\gamma)} + z_4 e^{{\rm i}\gamma}| + |z_2e^{{\rm i}\alpha} + \bar z_3 e^{-{\rm i} \beta}| \\
& = |z_1 e^{{\rm i}(\alpha+\beta)} + z_4| + |z_2 e^{{\rm i}(\alpha+\beta)} + \bar z_3| \\
& \le \max_\theta |z_1 e^{{\rm i} \theta}+z_4| + |z_2 e^{{\rm i} \theta} + \bar z_3|.
\end{aligned}
$$
For the converse, we put $\phi := -\arg (z_1 e^{{\rm i}\theta}+z_4)$ and $\psi:=-\arg (z_2e^{{\rm i}\theta}+\bar z_3)$.
Then, we have
$$
\begin{aligned}
|z_1 e^{{\rm i} \theta}+z_4| + |z_2 e^{{\rm i} \theta} + \bar z_3|
& = {\rm Re}\left((z_1 e^{{\rm i} \theta}+z_4) e^{{\rm i}\phi}\right) + {\rm Re}\left((z_2 e^{{\rm i} \theta}+\bar z_3) e^{{\rm i}\psi}\right) \\
& = {\rm Re} \left(z_1 e^{{\rm i}(\theta+\phi)} + z_2 e^{{\rm i}(\theta+\psi)} + z_3 e^{-{\rm i}\psi} + z_4 e^{{\rm i}\phi} \right)\\
& \le C(z),
\end{aligned}
$$
where the last inequality follows from $(\theta+\psi)-\psi+\phi = \theta+\phi$.
\end{proof}

\section{A proof of G\"uhne's conjecture}

In this section, we show that G\"uhne's necessary condition {\rm (\ref{Guhne})} is also sufficient
for separability of GHZ diagonal states.
We begin with a characterization of separability of general {\sf X}-shaped states.

\begin{proposition}\label{general}
Let $\varrho$ be a three qubit state whose {\sf X}-part is given by $X(a,b,c)$. If $\varrho$ is separable, then the inequality
$$
2A(x,y) \mathcal L (\varrho, z) \le C(z) \left(\sum_{i=1}^4 a_i x_i + \sum_{i=1}^4 b_i y_i \right)
$$
holds for every $x, y \in \mathbb R^4_+$ and $z \in \mathbb C^4$.
The converse holds when $\varrho$ is {\sf X}-shaped.
\end{proposition}

\begin{proof}
We consider $W=X(s,t,u)$ with
$$
s={1 \over A(x, y)} x, \qquad t={1 \over A(x, y)} y, \qquad u=-{1 \over C(z)}(z_1, z_2, z_3, \bar z_4).
$$
By Proposition \ref{3to1}, we have
$$
\begin{aligned}
B(u) & = \max_\theta {|z_1 e^{{\rm i}\theta}+z_4|+|z_2e^{{\rm i}\theta}+\bar z_3| \over C(z)} \\
& = 1 \\
& = \inf_{r>0} {\sqrt{(x_1 r^{-1}+y_4r)(x_4r^{-1}+y_1r)} +\sqrt{(x_2r^{-1}+y_3r)(x_3r^{-1}+y_2r)} \over A(x, y)} \\
& = A(s,t).
\end{aligned}
$$
By Proposition \ref{XEW}, it follows that
$$
0 \le \lan \varrho, W \ran = {\sum_{i=1}^4 a_i x_i + \sum_{i=1}^4 b_i y_i \over A(x, y)}
- {2 \mathcal L (\varrho, z) \over C(z)}.
$$

For the converse, it suffices to show that $\lan \varrho, W \ran \ge 0$ for every {\sf X}-shaped witness $W$ by Proposition \ref{X-witness}.
We put $W=X(x, y, -(z_1,z_2,z_3,\bar z_4))$.
We may assume without loss of generality that $A(x, y) \ge 1 \ge B(z_1,z_2,z_3,\bar z_4)$ and $\mathcal L (\varrho, z) \ge 0$. Then we have
$$
\begin{aligned}
\lan \varrho, W \ran & = \sum_{i=1}^4 a_i x_i + \sum_{i=1}^4 b_i y_i - 2 \mathcal L (\varrho, z) \\
& \ge {\sum_{i=1}^4 a_i x_i + \sum_{i=1}^4 b_i y_i \over A(x, y)} - {2 \mathcal L (\varrho, z) \over C(z)}
\ge 0,
\end{aligned}
$$
as it was desired.
\end{proof}

Choosing special types of vectors $x$ and $y$ in Proposition \ref{general} gives rise to
the G\"uhne's necessary condition for separability. We include here our proof for completeness as well as
motivation to prove sufficiency. We recall that the {\sf X}-part of a separable state is again separable,
and the following criteria involves only the {\sf X}-part of a state.

\begin{theorem}[G\"uhne]\label{guhne}
If $\varrho$ is a three qubit separable state, then the inequality {\rm (\ref{Guhne})} holds
for every $z \in \mathbb C^4$.
\end{theorem}

\begin{proof}
It suffices to prove the required inequality,
when $c$ is a nonzero vector. By the PPT condition, all $a_i, b_i$ are nonzero.
Now, we fix $i$ among $1,2,3,4$, and define $x, y \in\mathbb R^4_+$ by
$$
x_j = \begin{cases} \sqrt{b_i \slash a_i}, & j=i \\ 0, & j \ne i \end{cases},  \qquad
y_j = \begin{cases} \sqrt{a_i \slash b_i}, & j=i \\ 0, & j \ne i \end{cases}.
$$
Then we have
$A(x, y) = 1$ and
$$
\sum_{j=1}^4 a_j x_j + \sum_{j=1}^4 b_j y_j = 2\sqrt{a_i b_i},
$$
which yields $\mathcal L (\varrho, z) \le C(z) \sqrt{a_i b_i}$.

Next, we define $x, y \in \mathbb R^4_+$ as follows:
$$
x =
\left( \sqrt[4]{a_1^{-3}b_2b_3a_4}, 0, 0,\ \sqrt[4]{a_1b_2b_3a_4^{-3}}\right),\quad
y =
\left( 0, \sqrt[4]{a_1b_2^{-3}b_3a_4}, \sqrt[4]{a_1b_2b_3^{-3}a_4}, 0\right).
$$
Then we have
$$
\begin{aligned}
&A(x, y) = \inf_{r>0} \left(r^{-1}\sqrt[4]{a_1^{-1}b_2b_3a_4^{-1}} + r \sqrt[4]{a_1b_2^{-1}b_3^{-1}a_4}\right) =2,\\
&
\sum_{j=1}^4 a_j x_j + \sum_{j=1}^4 b_j y_j = 4 \sqrt[4]{a_1b_2b_3a_4},
\end{aligned}
$$
which yields $\mathcal L (\varrho, z) \le C(z) \sqrt[4]{a_1b_2b_3a_4}$.
Finally, we use
$$
x =
\left(0, \sqrt[4]{b_1 a_2^{-3}a_3b_4}, \sqrt[4]{b_1 a_2a_3^{-3}b_4}, 0\right),\quad
y =
\left(\sqrt[4]{b_1^{-3} a_2a_3b_4}, 0, 0, \sqrt[4]{b_1 a_2 a_3b_4^{-3}}\right)
$$
in the exactly same way, to get the inequality $\mathcal L (\varrho, z) \le C(z) \sqrt[4]{b_1a_2a_3b_4}$.
\end{proof}

Now, we prove that the inequality  {\rm (\ref{Guhne})} is equivalent to separability of
GHZ diagonal states. The proof is similar to that of Proposition \ref{general} if we use
Theorem \ref{GHZ-witness} instead of Proposition \ref{X-witness}.

\begin{theorem}\label{main-sep}
Let $\varrho=X(a,a,c)$ be a three qubit GHZ diagonal state. Then, $\varrho$ is separable if and only if the
inequality {\rm (\ref{Guhne})} holds for every $z \in \mathbb R^4$.
\end{theorem}

\begin{proof}
It remains to prove the \lq if\rq\ part, and it
suffices to show that $\lan \varrho, W \ran \ge 0$ for every GHZ diagonal witness $W$ by Theorem \ref{GHZ-witness}.
We put $W=X(x,x,-z)$.
We may assume without loss of generality that $A(x, x) \ge 1 \ge C(z)$ and $\mathcal L (\varrho, z) \ge 0$.
We have
$$
\begin{aligned}
A(x, x) & = \inf_{r>0} \sqrt{(x_1 r^{-1}+x_4r)(x_4r^{-1}+x_1r)} +\sqrt{(x_2r^{-1}+x_3r)(x_3r^{-1}+x_2r)} \\
& = \inf_{r>0} \sqrt{x_1 x_4 (r^{-2}+r^2)+x_1^2+x_4^2} +\sqrt{x_2 x_3 (r^{-2}+r^2)+x_2^2+x_3^2} \\
& = x_1+x_2+x_3+x_4.
\end{aligned}
$$
Note that the minimums of two terms occur simultaneously when $r=1$.
Therefore, it follows that
$$
\begin{aligned}
\lan \varrho, W \ran & = 2 \sum_{i=1}^4 a_i x_i - 2 \mathcal L (\varrho, z) \\
& \ge 2 {\sum_{i=1}^4 a_i x_i \over A(x, x)} - 2 {\mathcal L (\varrho, z) \over C(z)} \\
& = 2 \sum_{i=1}^4 a_i {x_i \over x_1+x_2+x_3+x_4} - 2 {\mathcal L (\varrho, z) \over C(z)} \\
& \ge 2 \min \{a_1,a_2,a_3,a_4\} - 2 {\mathcal L (\varrho, z) \over C(z)}
 \ge 0.
\end{aligned}
$$
This completes the proof.
\end{proof}

\section{Characterization of separability by entries}

In order to characterize separability of three qubit GHZ
diagonal states in terms of entries, it suffices to find the maximum of the function
$$
f(z_1,z_2,z_3,z_4):=\frac{c_1z_1+c_2z_2+c_3z_3+c_4z_4}{C(z_1,z_2,z_3,z_4)}
=\frac{c_1z_1+c_2z_2+c_3z_3+c_4z_4}
{\max_\theta\left(|z_1e^{{\rm i}\theta} + z_4|+|z_2 e^{{\rm i}\theta} + z_3|\right)}
$$
of real variables,
by Theorem \ref{main-sep}.
We note that $f$ is a continuous function defined on the domain $\mathbb
R^4\setminus \{0\}$, and enjoys the relations
\begin{equation}\label{ruled}
f(\lambda z)=f(z)~~ \text{for each}~ \lambda>0 \qquad \text{and} \qquad f(-z)=-f(z).
\end{equation}
Hence, $f$ has the maximum which occurs on the compact set $\Delta=\{z: \sum_{i=1}^4|z_i|=1\}$.
In order to find this maximum, we first describe the function $C$ in terms of real variables $z_i$'s.

Because the function $C(z_1,z_2,z_3,z_4)$ is invariant under replacing two $z_i$'s
by $-z_i$, it suffices to consider two cases $z_1,z_2,z_3,z_4 \ge 0$ and $z_1<0$, $z_2,z_3,z_4>0$.
In the former case, it is trivial to see that $C(z)=\sum_{i=1}^4|z_i|$.
For the latter case,
we put $x=\cos \theta$ and define the function $g : [-1,1] \to \mathbb R$ by
$$
g(x)=\sqrt{z_1^2+z_4^2-2|z_1|z_4 x} + \sqrt{z_2^2+z_3^2+2z_2z_3 x},\qquad x\in [-1,1],
$$
whose maximum is $C(z)$. By one variable calculus, one may check that
$$
\begin{aligned}
g'(x) \ge 0 &\Longleftrightarrow z_2z_3\sqrt{z_1^2+z_4^2-2|z_1|z_4x} \ge |z_1|z_4\sqrt{z_2^2+z_3^2+2z_2z_3x}\\
&\Longleftrightarrow x \le {z_2^2z_3^2(z_1^2+z_4^2)-z_1^2z_4^2(z_2^2+z_3^2) \over 2|z_1|z_2z_3z_4(|z_1|z_4+z_2z_3)} =: \alpha.
\end{aligned}
$$
If $-1 < \alpha < 1$, then
$g$ takes the maximum at $\alpha$ with
$$
\begin{aligned}
C(z)
=g(\alpha)&= \left(\sqrt{{z_2z_3 \over |z_1|z_4}} + \sqrt{|z_1|z_4
\over
z_2z_3}\right)\sqrt{(|z_1|z_2+z_3z_4)(|z_1|z_3+z_2z_4)
\over |z_1|z_4+z_2z_3}\\
&=
\sqrt{\frac{(z_1z_2-z_3z_4)(z_2z_4-z_1z_3)(z_1z_4-z_2z_3)}{-z_1z_2z_3z_4}}.
\end{aligned}
$$
We also have the following:
$$
\alpha \le -1 \Longleftrightarrow {1 \over |z_1|} + {1 \over z_4} \le \left|{1 \over z_2}-{1 \over z_3}\right|,
\qquad \alpha \ge 1 \Longleftrightarrow {1 \over z_2} + {1 \over z_3} \le \left|{1 \over |z_1|}-{1 \over z_4}\right|.
$$
In the first case, we have $C(z)=g(-1)=|z_1|+z_4+|z_2-z_3|$. In the second case,
we have $C(z)=g(1)=||z_1|-z_4|+z_2+z_3$.

We partition the domain of $f$ by $\mathbb R^4\setminus\{0\}=\Omega^+\sqcup\Omega^-$ with
$$
\Omega^+=\{z\in \mathbb R^4\setminus\{0\}: z_1z_2z_3z_4\ge 0\},\ \
\Omega^-=\{z\in \mathbb R^4\setminus\{0\}: z_1z_2z_3z_4< 0\}.
$$
We also partition $\Omega^-$ into the five regions $\Omega^-_0,\Omega^-_1,\Omega^-_2,\Omega^-_3,\Omega^-_4$:
$$
\begin{aligned}
\Omega^-_i&=\left\{ z\in\Omega^-: \textstyle{ \frac
1{|z_i|} \ge \sum_{j\neq i}\frac 1{|z_j|}}\right\},\qquad i=1,2,3,4,\\
\Omega^-_0&=\Omega^-\setminus\left(\sqcup_{i=1}^4\Omega^-_i\right)
=\left\{z\in\Omega^-: \textstyle{\frac 1{|z_i|} < \sum_{j\neq
i}\frac 1{|z_j|}},\ i=1,2,3,4\right\}.
\end{aligned}
$$
Note that $z\in\Omega^-$ belongs to $\Omega^-_0$ if and only if the numbers
$\frac 1{|z_i|}$ with $i=1,2,3,4$ make a quadrangle.
In the case of $z_1<0$ and $z_2,z_3,z_4>0$, we have $\alpha \ge 1$ for $z \in \Omega_1^-$ or
$z \in \Omega_4^-$, while $\alpha \le -1$ for $z \in \Omega_2^-$ or $z \in \Omega_3^-$. In each case, $C(z)$ is equal to
$$
-|z_1|+z_2+z_3+z_4,\quad |z_1|+z_2+z_3-z_4,\quad |z_1|-z_2+z_3+z_4,\quad |z_1|+z_2-z_3+z_4,
$$
respectively.
We also have
\begin{align}\label{-1<1}
-1<\alpha<1 &\Longleftrightarrow {1 \over |z_1|} + {1 \over |z_4|} > \left|{1 \over |z_2|}-{1 \over |z_3|}\right|,~
{1 \over |z_2|} + {1 \over |z_3|} > \left|{1 \over |z_1|}-{1 \over |z_4|}\right|\\
&\Longleftrightarrow z \in \Omega_0^- \nonumber.
\end{align}
We summarize as follows:

\begin{proposition}\label{C(z)}
For $z \in \mathbb R^4\setminus\{0\}$, we have
$$
C(z)=
\begin{cases}
|z_1|+|z_2|+|z_3|+|z_4|, \qquad  &z\in\Omega^+, \\
|z_1|+|z_2|+|z_3|+|z_4|-2|z_i|,\qquad &z\in \Omega^-_i\ (i=1,2,3,4),\\
\sqrt{\frac{(z_1z_2-z_3z_4)(z_2z_4-z_1z_3)(z_1z_4-z_2z_3)}{-z_1z_2z_3z_4}},
\qquad & z\in \Omega^-_0.
\end{cases}
$$
\end{proposition}

We also partition $\Omega^-$ into the eight domains $\Omega^\sigma$, where $\sigma$ is one of:
$$
-+++,\
+-++,\
++-+,\
+++-,\
+---,\
-+--,\
--+-,\
---+,
$$
according to the signs of $z_1,z_2,z_3$ and $z_4$. We put
$$
\Omega^\sigma_i:=\Omega^-_i\cap\Omega^\sigma,\qquad i=0,1,2,3,4.
$$
In order to understand the global behavior of the function $f$,
we consider the boundary
$$
B_i^\sigma := \left\{z \in \Omega^\sigma : {1 \over |z_i|} = \sum_{j \ne i} {1 \over |z_j|} \right\}
$$
between $\Omega^\sigma_0$ and $\Omega^\sigma_i$
for $i=1,2,3,4$.
For example, $B_1^{-+++}$
consists of $z\in\mathbb R^4$ satisfying
$$
z_1<0,\ z_2,z_3,z_4>0,\qquad z_1=-\frac{z_2z_3z_4}{z_2z_3+z_3z_4+z_4z_2}.
$$
If we fix two of $z_2,z_3,z_4$ and take another $z_i\to 0^+$, then we
have $z_1\to 0^-$. This means that every point on the plane
$z_1=z_2=0$, $z_1=z_3=0$ or $z_1=z_4=0$ with $z_i\ge 0\
(i=2,3,4)$ is a limit point of $B^{-+++}_1$.

By the relation (\ref{ruled}), we may restrict the domain of $f$ on the set $\Delta$.
Then the intersection $\Omega^\sigma\cap\Delta$ with $\Omega^\sigma$ is a tetrahedron,
and $B^\sigma_i\cap\Delta$ is a two-dimensional surface in the tetrahedron.
For example, the set $\Omega^{-+++} \cap \Delta$ is the tetrahedron with four
extreme points $-E_1, E_2, E_3$ and $E_4$, and the surface $B^{-+++}_1\cap\Delta$ consists of
$z\in\mathbb R^4$ satisfying
$$
z_1<0,\ z_2,z_3,z_4>0,\ {1 \over z_1}+{1 \over z_2}+{1 \over z_3}+{1 \over z_4}=0,\ -z_1+z_2+z_3+z_4=1.
$$
Limit points of the surface
include the three edges of the face $F_1$ with extreme points
$E_2,E_3,E_4$. This surface is near from the face $F_1$, and far from
the extreme point $-E_1$. The maximum distance from the surface $B^{-+++}_1\cap\Delta$ to
the face $F_1$ is given by
$$
{\textstyle{\|(-\frac 1{10},\frac 3{10}, \frac 3{10}, \frac 3{10})-(0,\frac 13,\frac 13,\frac 13)\|}}
=\frac 1 {5\sqrt 3}.
$$
The each piece of the region $\Omega^\sigma_0$ in $\Delta$ occupies the central part of
the corresponding tetrahedron, and touches all the edges and extreme points of the tetrahedron, but
is apart from the four faces with dimension two.

Since $f$ is a fraction of linear functions
on each $\Omega^\sigma_i$ with $i=1,2,3,4$, all its partial derivatives never vanish on them, and so $f$
has no extreme value on the interior of $\sqcup_{i=1}^4\Omega^-_i$.
Therefore, we see that the maximum value of $f$ occurs on
$\Omega^+$, $B_i^\sigma$ or $\Omega^-_0$. If the maximum occurs on $\Omega^+$
then we have $\max f(z)=\max_i|c_i|$, in which case the separability of $\varrho$ is equivalent to
the condition of PPT by Theorem \ref{main-sep}.

Now, we consider the function $f(z)$ with the third expression of $C(z)$, and
suppose that $s=(s_1,s_2,s_3,s_4)\in\Omega^-_0$ is a critical point of $f$.
With the notation $t_i=1 \slash s_i$, we have
\begin{equation}\label{tttt}
\begin{aligned}
t_1&=c_1(-c_1^2+c_2^2+c_3^2+c_4^2)-2c_2c_3c_4,\\
t_2&=c_2(+c_1^2-c_2^2+c_3^2+c_4^2)-2c_1c_3c_4,\\
t_3&=c_3(+c_1^2+c_2^2-c_3^2+c_4^2)-2c_1c_2c_4,\\
t_4&=c_4(+c_1^2+c_2^2+c_3^2-c_4^2)-2c_1c_2c_3,\\
\end{aligned}
\end{equation}
up to nonzero scalar multiplications. See Appendix for the details. One may also check
$$
\begin{aligned}
f(s)^2&={4(c_1c_2-c_3c_4)(c_1c_3-c_2c_4)(c_1c_4-c_2c_3) \over (c_1+c_2+c_3+c_4)(c_1+c_2-c_3-c_4)(c_1-c_2-c_3+c_4)(c_1-c_2+c_3-c_4)}\\
&=\frac{(\lambda_5\lambda_6+\lambda_7\lambda_8)(\lambda_5\lambda_7+\lambda_6\lambda_8)(\lambda_5\lambda_8+\lambda_6\lambda_7)}
{8^2\lambda_5\lambda_6\lambda_7\lambda_8} 
\end{aligned}
$$
which appears in the sufficient condition of G\"uhne (\ref{sufficient}).
We are going to look for the condition with which the critical point $s=(s_1,s_2,s_3,s_4)$ belongs to
$\Omega^-_0$. To do this, we note by calculation the following relations between $t_i$ and $\lambda_j$:
$$
\begin{aligned}
(t_2+t_3)^2-(t_1+t_4)^2&={1 \over 2^6}\lambda_5\lambda_8(\lambda_6\lambda_7)^2,\\
(t_1-t_4)^2-(t_2-t_3)^2&={1 \over 2^6}\lambda_6\lambda_7(\lambda_5\lambda_8)^2.
\end{aligned}
$$
Since $s=(s_1,s_2,s_3,s_4)\in\Omega^-$, we have $(t_1t_4)(t_2t_3)<0$.
We first consider the case $t_1t_4>0$ and $t_2t_3<0$. In this case, we have
$$
\begin{aligned}
|t_2|+|t_3| > ||t_1|-|t_4||\ &\Longleftrightarrow\
(t_2-t_3)^2 > (t_1-t_4)^2\ \Longleftrightarrow\
\lambda_6\lambda_7 < 0,\\
|t_1|+|t_4| > ||t_2|-|t_3||\ &\Longleftrightarrow\
(t_1+t_4)^2 > (t_2+t_3)^2\ \Longleftrightarrow\
\lambda_5\lambda_8 < 0.
\end{aligned}
$$
In case of $t_1t_4<0$ and $t_2t_3>0$, we also have
$$
\begin{aligned}
|t_2|+|t_3| > ||t_1|-|t_4||\ &\Longleftrightarrow\
(t_2+t_3)^2 > (t_1+t_4)^2\ \Longleftrightarrow\
\lambda_5\lambda_8 > 0,\\
|t_1|+|t_4| > ||t_2|-|t_3||\ &\Longleftrightarrow\
(t_1-t_4)^2 > (t_2-t_3)^2\ \Longleftrightarrow\
\lambda_6\lambda_7 > 0.
\end{aligned}
$$
By (\ref{-1<1}), we see that $s\in\Omega^-_0$ if and only if
$$
t_1t_4>0,\ t_2t_3<0,\ \lambda_6\lambda_7 < 0,\ \lambda_5\lambda_8 < 0\ \
{\text{\rm or}}\ \
t_1t_4<0,\ t_2t_3>0,\ \lambda_6\lambda_7 > 0,\ \lambda_5\lambda_8 > 0
$$
if and only if
$\lambda_5\lambda_6\lambda_7\lambda_8 > 0$
and
\begin{equation}\label{cond-entry}
t_1t_4\lambda_6\lambda_7 < 0 \quad {\text{\rm and}}\quad t_2t_3\lambda_5\lambda_8 > 0.
\end{equation}
Summing up, we have seen that
the function $f$ with the third expression of $C(z)$ has a critical value on  $\Omega^-_0$
if and only if the inequality
$\lambda_5\lambda_6\lambda_7\lambda_8 > 0$ together with the condition {\rm (\ref{cond-entry})}
holds. Now, we are ready to state and prove the following, which correct a result in \cite{chen_qip_2015}.

\begin{theorem}\label{main_condition}
Let $\varrho=X(a,a,c)$ be a GHZ diagonal state. Suppose that $\lambda_5,\lambda_6,\lambda_7,\lambda_8$ and
$t_1,t_2,t_3,t_4$ are given by {\rm (\ref{lambdas})} and {\rm (\ref{tttt})}, respectively. Then we have the following:
\begin{enumerate}
\item[(i)]
if $\lambda_5\lambda_6\lambda_7\lambda_8 \le 0$,
then $\varrho$ is separable if and only if it is of PPT;
\item[(ii)]
if $\lambda_5\lambda_6\lambda_7\lambda_8 > 0$ and {\rm (\ref{cond-entry})} does not hold, then $\varrho$ is separable if and only if
it is of PPT;
\item[(iii)]
if $\lambda_5\lambda_6\lambda_7\lambda_8 > 0$ and {\rm (\ref{cond-entry})} holds, then $\varrho$ is separable if and only if
the inequality {\rm (\ref{sufficient})} holds.
\end{enumerate}
\end{theorem}

\begin{proof}
We first consider the behavior of the function $f$ on the surface $B^\sigma_i\cap\Delta$, say $B^{-+++}_1\cap\Delta$.
To do this, we denote by 
$f_2$ and $f_3$ the function $f$ with the second and third expression of $C(z)$ in Proposition \ref{C(z)}.
We take the line segment
$$
z(t)=(1-t)(-E_1)+t(0,\omega_2,\omega_3,\omega_4)=(t-1,t\omega_2,t\omega_3,t\omega_4),\qquad 0\le t\le 1
$$
between the extreme point $-E_1$ and a point $(0,\omega_2,\omega_3,\omega_4)$ on the opposite face $F_1$,
with $\omega_2+\omega_3+\omega_4=1$ and $\omega_i\ge 0$.
The line segment meets the surface $B^{-+++}_1\cap\Delta$ at
$$
t_0=\frac{\omega_2\omega_3+\omega_3\omega_4+\omega_4\omega_2}{\omega_2\omega_3\omega_4+\omega_2\omega_3+\omega_3\omega_4+\omega_4\omega_2}.
$$
One may check that the derivatives of $f_2(z(t))$ and $f_3(z(t))$ at $t=t_0$
are positive scalar multiples of $c_1-c_2\omega_2-c_3\omega_3-c_4\omega_4$.
Therefore, we see that
the function $f$ has no extreme value on the surface $B^{-+++}_1\cap\Delta$ possibly except on the curve
$$
\{(t_0-1,t_0\omega_2,t_0\omega_3,t_0\omega_4) : c_2\omega_2+c_3\omega_3+c_4\omega_4=c_1,~ \omega_2+\omega_3+\omega_4=1, \omega_i \ge 0\},
$$
where $f$ has the constant value $c_1$. The other cases can be handled in the same way to see that possible extreme values are
$\pm c_i$.

Suppose that the assumption of (i) or (ii) holds.
Then we see by the above argument that the global maximum of $f$ is $\max|c_i|$.
Therefore, we have the required results by Theorem \ref{main-sep}.
In case of (iii), we know that the maximum of $f$ is $\max|c_i|$ or the number in the inequality (\ref{sufficient}).
One may check that the square of this number is equal to $c_i^2 +{2^4 t_i^2 \over \lambda_5\lambda_6\lambda_7\lambda_8}$
for each $i=1,2,3,4$. Therefore, we conclude that the maximum of $f$
is the number in (\ref{sufficient}), when $\lambda_5\lambda_6\lambda_7\lambda_8>0$.
This completes the proof by Theorem \ref{main-sep} again.
\end{proof}

The case (i) is just Kay's criterion. The case (iii) tells us that G\"uhne's sufficient condition
is also necessary under the condition (\ref{cond-entry}).
The case (ii) is most interesting. This case shows that separability of $\varrho$ may be equivalent
to PPT even though the Kay's condition $\lambda_5\lambda_6\lambda_7\lambda_8< 0$ is not satisfied.
This case also provides examples of separable states violating G\"uhne's sufficient condition.

Motivated by Kay's example in \cite{kay_2011}, we consider the family of GHZ diagonal states
$$
\varrho_{p,q}=\frac 18X({\bf 1},{\bf 1}, (p,p,q,p))
$$
where ${\bf 1}=(1,1,1,1)$. This is really a state if and only if $\max\{|p|,|q|\}\le 1$
if and only if it is of PPT. In this case, we have
$$
\begin{aligned}
&t_1=t_2=t_4={1 \over 8^3}p(p-q)^2,\quad t_3=-{1 \over 8^3}(2p+q)(p-q)^2,\\
&\lambda_5={1 \over 4}(3p+q),\quad \lambda_6=\lambda_8={1 \over 4}(-p+q),\quad \lambda_7={1 \over 4}(p-q).
\end{aligned}
$$
Therefore, we have the following:
\begin{itemize}
\item
the case (i) occurs if and only if $(p-q)(3p+q) \le 0$,
\item
the case (ii) occurs if and only if $(p-q)(3p+q) > 0$ and $p(2p+q) \le 0$,
\item
the case (iii) occurs if and only if $(p-q)(3p+q) > 0$ and $p(2p+q) > 0$,
\item
the inequality {\rm (\ref{sufficient})} holds if and only if $(p-q)(3p+q)> 0$ and $\frac{4p^3}{3p+q}\le 1$.
\end{itemize}
See Figure 1.

\begin{figure}[t]
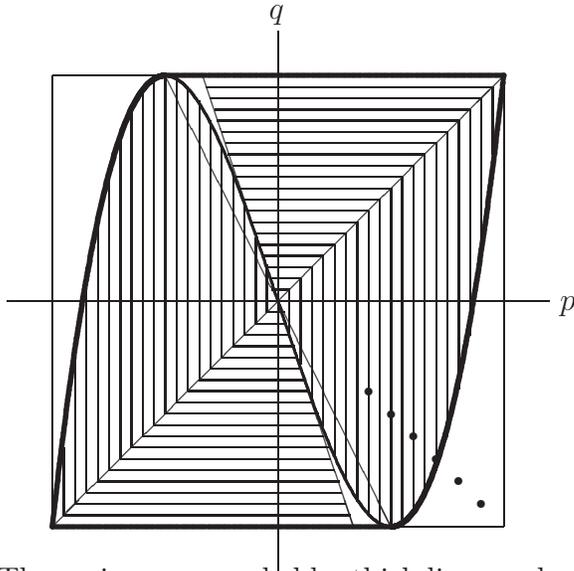

\centering
\input GHZ_figure.tex

\end{picture}
\end{center}
\caption{The region surrounded by thick lines and curves denotes separable states
among $\varrho_{p,q}$.
States outside of this curves in the box are PPT entanglement.
Especially, dots denote Kay's examples of one parameter family to find PPT entanglement among GHZ diagonal states.
Kay's criterion works on the parts with horizontal lines.
The parts with vertical lines represent separable states satisfying G\"uhne's sufficient condition.
States in remaining parts violate this condition even though they are separable and satisfy $\lambda_5\lambda_6\lambda_7\lambda_8>0$.
}
\end{figure}

In multi-partite systems, we note that there are various notions
of separability and entanglement
according to partitions of systems. For example, we say that a multi-partite state is fully bi-separable if it is
separable as a bi-partite state with respect to every bi-partition of systems.
Characterization of full bi-separability and corresponding witnesses are given for general multi-qubit {\sf X}-shaped states
\cite{{han_kye_EW}}. We note that only absolute values of anti-diagonal entries involved in these criteria.
On the other hand, the arguments of anti-diagonal entries play a key role in the characterization
of witnesses corresponding (full) separability, as we have seen in Proposition \ref{XEW}.
In this regard, it would be very interesting to get a characterization for separability
of general {\sf X}-shaped three qubit states when the anti-diagonals have complex entries.

\section{Appendix}

Suppose that $s=(s_1,s_2,s_3,s_4)$ is a critical point of the function
$$
f(z_1,z_2,z_3,z_4)=
\frac{c_1z_1+c_2z_2+c_3z_3+c_4z_4}
{\sqrt{\frac{(z_1z_2-z_3z_4)(z_2z_4-z_1z_3)(z_1z_4-z_2z_3)}{-z_1z_2z_3z_4}}}
$$
when $z\in \Omega^-_0$. From
$$
\log C(z) = {1 \over 2}(\log|z_1z_2-z_3z_4| + \log |z_2z_4-z_1z_3| + \log |z_1z_4-z_2z_3|-\log(-z_1z_2z_3z_4)),
$$
we get logarithmic derivatives
$$
\begin{aligned}
2{{\partial C \over \partial z_1} (z) \over C(z)} &= +{z_2 \over z_1z_2-z_3z_4}-{z_3 \over z_2z_4-z_1z_3}+{z_4 \over z_1z_4-z_2z_3}-{1 \over z_1},\\
2{{\partial C \over \partial z_2} (z) \over C(z)} &= +{z_1 \over z_1z_2-z_3z_4}+{z_4 \over z_2z_4-z_1z_3}-{z_3 \over z_1z_4-z_2z_3}-{1 \over z_2},\\
2{{\partial C \over \partial z_3} (z) \over C(z)} &= -{z_4 \over z_1z_2-z_3z_4}-{z_1 \over z_2z_4-z_1z_3}-{z_2 \over z_1z_4-z_2z_3}-{1 \over z_3},\\
2{{\partial C \over \partial z_4} (z) \over C(z)} &= -{z_3 \over z_1z_2-z_3z_4}+{z_2 \over z_2z_4-z_1z_3}+{z_1 \over z_1z_4-z_2z_3}-{1 \over z_4}.
\end{aligned}
$$
We also consider the logarithmic derivatives of $f$ at $s$ as
$$
0={{\partial f \over \partial z_i} (s) \over f(s)} = {c_i \over c_1s_1+c_2s_2+c_3s_3+c_4s_4}-{{\partial C \over \partial z_i} (s) \over C(s)}.
$$
Combining the above, we obtain simultaneous equations
\begin{align}
{2c_1 s_1 \over c_1s_1+c_2s_2+c_3s_3+c_4s_4}&=+{s_1s_2 \over s_1s_2-s_3s_4}-{s_1s_3 \over s_2s_4-s_1s_3}+{s_1s_4 \over s_1s_4-s_2s_3}-1, \label{1}\\
{2c_2 s_2 \over c_1s_1+c_2s_2+c_3s_3+c_4s_4}&=+{s_1s_2 \over s_1s_2-s_3s_4}+{s_2s_4 \over s_2s_4-s_1s_3}-{s_2s_3 \over s_1s_4-s_2s_3}-1, \label{2}\\
{2c_3 s_3 \over c_1s_1+c_2s_2+c_3s_3+c_4s_4}&=-{s_3s_4 \over s_1s_2-s_3s_4}-{s_1s_3 \over s_2s_4-s_1s_3}-{s_2s_3 \over s_1s_4-s_2s_3}-1, \label{3}\\
{2c_4 s_4 \over c_1s_1+c_2s_2+c_3s_3+c_4s_4}&=-{s_3s_4 \over s_1s_2-s_3s_4}+{s_2s_4 \over s_2s_4-s_1s_3}+{s_1s_4 \over s_1s_4-s_2s_3}-1 \label{4}.
\end{align}
Taking $(\ref{1})+(\ref{2}) \over (\ref{3})+(\ref{4})$, $(\ref{1})+(\ref{3}) \over (\ref{2})+(\ref{4})$, $(\ref{1})+(\ref{4}) \over (\ref{2})+(\ref{3})$, we have
$$
{c_1s_1+c_2s_2 \over c_3s_3+c_4s_4}=-{s_1s_2 \over s_3s_4}, \qquad {c_1s_1+c_3s_3 \over c_2s_2+c_4s_4}=-{s_1s_3 \over s_2s_4}, \qquad {c_1s_1+c_4s_4 \over c_2s_2+c_3s_3}=-{s_1s_4 \over s_2s_3},
$$
equivalently the system of linear equations
$$
\begin{aligned}
c_2\omega_1+c_1\omega_2+c_4\omega_3+c_3\omega_4&=0,\\
c_3\omega_1+c_4\omega_2+c_1\omega_3+c_2\omega_4&=0,\\
c_4\omega_1+c_3\omega_2+c_2\omega_3+c_1\omega_4&=0
\end{aligned}
$$
if we put
$$
\omega_1:=s_2s_3s_4,\quad \omega_2:=s_1s_3s_4,\quad \omega_3:=s_1s_2s_4,\quad \omega_4:=s_1s_2s_3.
$$
We take
$$
\omega_4=-\begin{vmatrix} c_2&c_1&c_4\\c_3&c_4&c_1\\c_4&c_3&c_2 \end{vmatrix}=c_4(-c_1^2-c_2^2-c_3^2+c_4^2)+2c_1c_2c_3.
$$
considering the relation (\ref{ruled}). Then, we have
$$
\begin{aligned}
\omega_1&=\begin{vmatrix} c_3&c_1&c_4\\c_2&c_4&c_1\\c_1&c_3&c_2 \end{vmatrix}=c_1(+c_1^2-c_2^2-c_3^2-c_4^2)+2c_2c_3c_4,\\
\omega_2&=\begin{vmatrix} c_2&c_3&c_4\\c_3&c_2&c_1\\c_4&c_1&c_2 \end{vmatrix}=c_2(-c_1^2+c_2^2-c_3^2-c_4^2)+2c_1c_3c_4,\\
\omega_3&=\begin{vmatrix} c_2&c_1&c_3\\c_3&c_4&c_2\\c_4&c_3&c_1 \end{vmatrix}=c_3(-c_1^2-c_2^2+c_3^2-c_4^2)+2c_1c_2c_4.
\end{aligned}
$$
Hence, (\ref{tttt}) follows from the relation
$$
-{1 \over s_1s_2s_3s_4}s_i=-{1 \over \omega_i}.
$$
Here, we may ignore the common positive scalar multiple $-1 \slash s_1s_2s_3s_4$ by (\ref{ruled}).

\end{document}